\documentclass[11pt]{article}

\makeatletter%
\@ifclassloaded{lipics-v2021}%
{\newenvironment{headdata}{\maketitle}{}
\usepackage[bold,full]{complexity}
\def\lipics
}%
{\usepackage{mypaper}
\usepackage[dvips,pdfstartview=FitH,pdfpagemode=UseNone,colorlinks=true,citecolor=blue,linkcolor=blue]{hyperref} 
}%
\usepackage{verbatim}

\usepackage{algorithm}
\usepackage{algorithmic}

\ifdefined\nat
\else
\newcommand{\nat}{\ensuremath{\mathbb{N}}}
\fi

\title{$P\not=NP$ relative to a $P$-complete oracle}
\makeatletter%
\ifdefined\lipics
\titlerunning{$\P^\P\not=\NP^\P$}

\author{Reiner Czerwinski}{TU Berlin (Alumnus), Berlin, Germany}{reiner.czerwinski@posteo.de}{https://orcid.org/0000-0002-4523-4420}{}
 \authorrunning{Reiner Czerwinski}
\Copyright{Reiner Czerwinski}

\else
{\author{Reiner Czerwinski\\TU Berlin (Alumnus)\\reiner.czerwinski@posteo.de}
\fi

\begin{document}
\begin{headdata}
  \begin{abstract}
    The $\P$ versus $\NP$ problem is still unsolved.
    But there are several oracles with $\P$ unequal $\NP$ relative
    to them.

    Here we will prove, that $\P\not=\NP$ relative to a $\P$-complete
    oracle.

    In this paper, we use padding arguments as the proof method. The padding arguments
    are not bounded by a computable function.
    Such as we can use methods from
    computability theory to separate complexity classes.
  \end{abstract}
\end{headdata}

\section{Introduction}

There are several oracles for which the relative P versus NP problem is
already solved. Baker, Gill, and Solovay\cite{baker1975relativizations} constructed an oracle
$A$ with $\P^A\not=\NP^A$ and an oracle $B$ with
$\P^B=\NP^B$.
Relative to a random oracle, $\P$ is unequal $\NP$\cite{Hstad2017AnAD}.
In this paper, we analyze P versus NP relative to a
P-complete oracle, i.e., we look at whether $\P^\P = \NP^\P$
or not.

\newcommand{\UNP}{\ensuremath{U_{\NP}}}
\newcommand{\oUP}{\ensuremath{U_{\P}}}
\newcommand{\DNP}{\ensuremath{D_{\NP}}}
\newcommand{\oDP}{\ensuremath{D_{\P}}}

As an oracle, we use the set
\[
  \oUP = \{\langle M, x, 1^t \rangle \mid \text{ TM }M\text{ accepts }x\text{ within }t\text{ steps}  \}
\]

To show, that $\NP^{\oUP}\not=\P^{\oUP}$ we use an \NP-complete set
\[
  \UNP =  \{\langle M, 1^n, 1^t \rangle \mid \exists x\in\{0,1\}^n   \text{ where TM }M\text{ accepts }x\text{ within }t\text{ steps} \} 
\]

In section \ref{sec:without} we will analyze the sets without time limits.
In this case, we can use methods from computability theory.

In section \ref{sec:withtime} we will show a connection from the sets in section \ref{sec:without} to the sets $\UNP$ and $\oUP$. So, we can conclude to $\NP^{\oUP}$.
\section{Computation without Time limit} \label{sec:without}
To analyze, whether $\UNP^{\oUP}$ is efficiently computable, we define
two sets analog to $\UNP$ and $\oUP$
without time limit.
\[\begin{split}
    \oDP = \{\langle M, x \rangle \mid \text{ TM }M\text{ accepts }x \} \\
    \DNP = \{    \{\langle M, 1^n \rangle \mid \exists x\in\{0,1\}^n   \text{ where TM }M\text{ accepts }x \} 
  \end{split}
\]
So, $\langle M, x \rangle\in\oDP$, if $x\in L(M)$ and
$\langle M, 1^n \rangle\in\DNP$, if $\exists x\in\{0,1\}^n$ with $x\in L(M)$.

It is obvious, that
\begin{itemize}
\item $\langle M, x \rangle\in\oDP$  iff $\exists t : \langle M, x, 1^t \rangle \in \oUP$
\item $\langle M, 1^n \rangle\in\DNP$  iff $\exists t : \langle M, 1^n, 1^t \rangle \in \UNP$
\end{itemize}
\begin{lemma} \label{uncompute}
  If $\langle M, \{0,1\}^n \rangle \in\DNP$, then it is undecidable,
  which $x\in\{0,1\}^n$ is in $L(M)$. 
\end{lemma}
\begin{proof}
 Let $M$ be an arbitrary TM. Due to Rice's theorem \cite{RICE} every nontrivial
  property of the language $L(M)$ is undecidable.
  Let $x \in \{0,1\}^n$ arbitrary but fixed. Whether $x\in L(M) $ is undecidable.
  Also, if $x$ is the only element in $\bigl(\{0,1\}^n \cap L(M)\bigr) $, is
  undecidable, because it is a nontrivial property of $L(M)$. So,
  it is undecidable, which element with length $n$ is in $L(M)$.
\end{proof}

\begin{corollary}\label{blackbox}
  Let $M$ be an arbitrary TM and $n\in\nat$.
  To test whether $\langle M, 1^n \rangle \in \DNP$ one needs
  $\theta(2^n)$ accesses to the oracle $\oDP$ in the worst case.
\end{corollary}
\begin{proof}
  Due to lemma \ref{uncompute} it is undecidable for an arbitrary
  TM $M$, whether an input $x$ with length $n$ is in $L(M)$.
  But the characteristic function of $L(M)$:
   \[ \chi_{L(M)}(x)=\begin{cases} 1 & x\in L(M) \\
                              0  & \text{otherwise}
                             \end{cases}
                           \]

     is computable with oracle $\oDP$ as a black box. So,
     the black box complexity is exponential.                      
  \end{proof}

  \section{Main Result}\label{sec:withtime}
  In this section we conclude from $\DNP$ with oracle $\oDP$
  to $\UNP$ with oracle $\oUP$.
  \begin{theorem}
    If $M$ is an arbitrary TM, then one needs  $\theta(2^n)$
    accesses to the oracle $\oUP$ to decide whether
    $\langle M, 1^n, 1^t \rangle \in \UNP$ in the worst case.
  \end{theorem}
  \begin{proof}
    Let $M$ be an arbitrary TM and $n\in\nat$.
    The set $\{0,1\}^n$ is finite. So, there is a
    $T\in\nat$ with
    \[  \langle M, 1^n, 1^T \rangle \in \UNP \iff \langle M, 1^n \rangle \in \DNP \]

    In this case we know for $x\in\{0,1\}^n$:
    $\langle M, 1^n, 1^T \rangle \in \oUP$ iff
    $ \langle M, x \rangle \in \oDP$. So, the calculation
    of $UNP$ with oracle $\oUP$
    has the same time complexity as the calculation
    of $DNP$ with oracle $\oDP$. Due to lemma \ref{blackbox}
    this is $\theta(2^n)$.

    Maybe one could argue, that there is a faster calculation than a black box
    search when the time limit is shorter.
    Assuming there is a faster calculation for
    $\langle M,1^n, 1^t  \rangle \in^? \UNP $, when
    \begin{equation}
      \exists x\in\{0,1\}^n \text{ with } \langle M, x, 1^t  \rangle \in \oUP
      \text{ but } x\in L(M) \text{.}
    \end{equation}
    Then it would be decidable, if $M$ accepts an $x$ in more than $t$ steps.
    But it is not.
  \end{proof}
  \section{Conclusion}
  Maybe $\P^\P=\P$ and $\NP^\P=\NP$. In this case the result would be
  a solution of the P vs. NP problem. But for solving P vs. NP one
  needs a new insight, as Oded Goldreich mentioned on his web page\cite{odedPvsNP}.
  Is there a new insight in this paper?
Someone anonymous told me:
  "Experts agree that you cannot solve problems like P versus NP with computability tricks". So, it is a novel insight.
  \bibliography{lit}{}
\bibliographystyle{plain}

\end{document}